\RequirePackage{fix-cm}
\documentclass[smallextended]{svjour3}       % onecolumn (second format)
\smartqed  % flush right qed marks, e.g. at end of proof
\newcounter{ContadorFazer}
\usepackage{xcolor}

\usepackage{enumerate}
\usepackage{graphicx}
\usepackage{amssymb}
\usepackage{complexity}
\usepackage{amsmath}
\usepackage{tikz}
\usepackage[ruled,vlined,linesnumbered]{algorithm2e}
\usepackage{hyperref}

\usepackage{cite}

\makeatletter
\renewcommand{\@biblabel}[1]{(#1)}
\makeatother

\newtheorem{choice}{Choice}

\newcommand{\semborda}[1]{\overset{\leftarrow}{#1}}
\newcommand{\borda}[1]{\overset{\rightarrow}{#1}}

\begin{document}

\title{Computing the hull number in toll convexity\thanks{Partially supported by CNPq, Brazil.}}


\author{Mitre C. Dourado}


\institute{Mitre C. Dourado \at
        	Instituto de Matemática, \\
        	Universidade Federal do Rio de Janeiro, Rio de Janeiro, Brazil. \\
            \email{mitre@dcc.ufrj.br}           %  \\
}

\date{Received: date / Accepted: date}

\maketitle

\begin{abstract}
A tolled walk $W$ between vertices $u$ and $v$ in a graph $G$ is a walk in which $u$ is adjacent only to the second vertex of $W$ and $v$ is adjacent only to the second-to-last vertex of $W$. A set $S \subseteq V(G)$ is toll convex if the vertices contained in any tolled walk between two vertices of $S$ are contained in $S$. The toll convex hull of $S$ is the minimum toll convex set containing~$S$. The toll hull number of $G$ is the minimum cardinality of a set $S$ such that the toll convex hull of $S$ is $V(G)$. The main contribution of this work is a polynomial-time algorithm for computing the toll hull number of a general graph.
\end{abstract}

%%%%%%%%%%%%%%%%%%%%%%%%%%%%%%%%%%%%%%%%%%%%%%%%%%%%%%%%%%%%%%%%%%%%%%%%%%%%%%%%%%%%
\section{Introduction} \label{sec:int}

A family ${\cal C}$ of subsets of a finite set $X$ is a {\em convexity on $X$} if $\varnothing, X \in {\cal C}$ and ${\cal C}$ is closed under intersection~\cite{Vel1993}.
Graph convexities have gained attention in the last decades~\cite{Duchet-mono,Gimbel2003,Henning2013,pelayo2013}.
Unlike the most studied graphs convexities, which are defined by a family of paths, the toll convexity~\cite{Alcon2015,GR2017} uses a special kind of walks.
Recall that a {\em walk} between vertices $u$ and $v$ of a graph $G$ (or a {\em $(u,v)$-walk}) is a sequence of vertices $w_1 \ldots w_k$ such that $k \ge 1$, $u = w_1$, $v = w_k$, and $w_iw_{i+1} \in E(G)$ for $1 \le i < k$.
A {\em tolled walk between vertices $u$ and $v$} (or a {\em tolled $(u,v)$-walk}) is a walk $w_1 \ldots w_k$ such that

\begin{itemize}
	\item $u \ne v$,
	\item $w_1w_i \in E(G)$ implies $i = 2$, and
	\item $w_iw_k \in E(G)$ implies $i = k-1$.
\end{itemize}

Given a graph $G$, a set $S \subseteq V(G)$ is {\em toll convex} (or {\em t-convex}) if the vertices contained in any tolled walk between two vertices of $S$ are contained in $S$; and $S$ is {\em toll concave} (or {\em t-concave}) if $V(G) - S$ is t-convex.
The {\em toll interval of $u,v \in V(G)$} is $[u,v] = \{w : w$ belongs to some tolled $(u,v)$-walk$\}$.
The {\em toll interval of $S$} is $[S] = \underset{u,v \in S}{\bigcup} [u,v]$ if $|S| \ge 2$ and $[S] = S$ otherwise.
If $[S] = V(G)$, then $S$ is said to be a {\em toll interval set of $G$} and the minimum cardinality of a toll interval set of $G$ is the {\em toll number of $G$}.
The {\em toll convex hull of $S$}, denoted by $\langle S \rangle$, is the minimum t-convex set containing $S$. If $\langle S \rangle = V(G)$, then $S$ is said to be a {\em toll hull set of $G$} and the minimum cardinality of a toll hull set of $G$ is the {\em toll hull number of $G$}. The main contribution of this work is a polynomial-time algorithm for computing the toll hull number of a general graph.

In the well-known {\em geodetic convexity}~\cite{FarberJamison1986,pelayo2013}, {\em monophonic convexity}~\cite{Duchet-mono,EJ1985}, and {\em $P_3$ convexity}~\cite{upper-Radon,Henning2013} all above concepts are analogously defined by replacing ``tolled walk'' by ``shortest path'', ``induced path'', and ``path of order three'', respectively. Note that every shortest path is an induced path, and every induced path is a tolled walk. This implies that the toll convex hull of $S$ contains the monophonic convex hull of $S$ and consequently also contains the geodetic convex hull of $S$. Therefore, the toll hull number is a lower bound of the monophonic and the geodetic hull numbers of the graph.

In the geodetic convexity, determining whether the hull number is at most $k$ is $\APX$-hard for general graphs~\cite{CDS2015}, $\NP$-complete for partial cube graphs~\cite{albenque} and chordal graphs~\cite{BDPR2018}, and solvable in polynomial time for unit interval graphs, cographs, split graphs~\cite{DGKPS2009}, cactus graphs, $P_4$-sparse graphs~\cite{araujoTCS}, distance hereditary graphs~\cite{KN2016}, ($P_5$,triangle)-free graphs~\cite{araujoENDM}.
In the $P_3$ convexity, this problem is $\APX$-hard even for bipartite graphs with maximum degree $\Delta \leq 4$~\cite{CDS2015}, and can be solved in polynomial time for block graphs and chordal graphs~\cite{Centeno2011}.
However, the monophonic hull number can be computed in polynomial time for general graphs~\cite{DPS-2010-monophonic}.
In the toll convexity, it is known that the hull number of every tree different of a caterpillar is equal to 2~\cite{Alcon2015}.

A graph $G$ is an {\em interval graph} if every vertex of $G$ can be associated with an interval of a straight line such that two vertices of $G$ are neighbors if and only if the corresponding intervals intersect. Given a convexity ${\cal C}$ on the vertex set of $G$, we say that $G$ is a {\em convex geometry under ${\cal C}$} if every ${\cal C}$-convex set of $G$ is equal to the ${\cal C}$-convex hull of its $\cal C$-extreme vertices.
In~\cite{Alcon2015}, it was shown that the interval graphs are precisely the graphs which are convex geometries in the toll convexity.
They also characterized the t-convex sets of a general graph and of some graph products.
In~\cite{GR2017}, the toll number of the Cartesian and the lexicographic product of graphs are studied, where some characterizations are presented.

The text is organized as follows. In Section~\ref{sec:tools}, we present useful definitions, notations, and results.
We also present the notion of hull characteristic family, which plays an important role in the proposed algorithm and can be an useful tool for future works dealing with the hull number. Section~\ref{sec:hullnumber} contains the algorithm for computing the toll hull number of a general graph in polynomial time. In the conclusions, we discuss that this result leads to an algorithm for generating all minimum toll hull sets of a general graph with polynomial delay and to a characterization of the toll extreme vertices of a graph.

%%%%%%%%%%%%%%%%%%%%%%%%%%%%%%%%%%%%%%%%%%%%%%%%%%%%%%%%%%%%%%%%%%%%%%%%%%%%%%%%%%%%
\section{Useful tools} \label{sec:tools}

We consider finite, simple, and undirected graphs. For a graph $G$, its vertex and edge sets are denoted by $V(G)$ and $E(G)$, for a vertex $w \in V(G)$, the open and the closed neighborhoods of $w$ are denoted by $N(w)$ and $N[w]$, respectively.
Vertices $u, v \in V(G)$ are {\em twins} if $N[u] = N[v]$.
For $S \subseteq V(G)$, the {\em neighborhood of $S$} is $N(S) = \left(\underset{u \in S}{\bigcup} N(u)\right) - S$, is $\borda{S} = \{u : u \in S$, and $N(u) \cap (V(G) - S) \neq \varnothing \}$. We will also use $\semborda{S} = S - \borda{S}$.
We can also use {\em $(u,v)$-path} to refer to a path between vertices $u$ and $v$.

Denote by $G - S$ the graph obtained by the deletion of the vertices of $S$; and by $G[S]$ the subgraph of $G$ induced by $S$. 
If every two vertices of $S$ are adjacent, then $S$ is a {\em clique of $G$}.
For a set $T \subseteq V(G) - S$, we say that $(S,T)$ is {\em complete} if every vertex of $S$ is adjacent to every vertex of $T$.
Vertex~$u$ is {\em simplicial in $G$} if $N(u)$ is a clique. If $V(G)$ is a clique, then $G$ is said to be a {\em complete graph}.

We say that $S \subset V(G)$ {\em separates vertices $u,v \in V(G)$} if there is a $(u,v)$-path in $G$ but there is no one in $G - S$; that $S$ is a separator of $G$ if $S$ separates some pair of vertices of $G$; and that $S \subset V(G)$ is a {\em clique separator of $G$} if $S$ is a clique and a separator of $G$.
We say that $G$ is {\em reducible} if it contains a clique separator, otherwise it is {\em prime}.
A {\em maximal prime subgraph of $G$} (or {\em mp-subgraph of $G$}) is a maximal induced subgraph of $G$ that is prime.
An mp-subgraph $F$ of a reducible graph $G$ is called {\em extremal} if there is an mp-subgraph $F'$ different of $F$ such that for every mp-subgraph $F''$ different of $F$, it holds $F \cap F'' \subseteq F \cap F'$. See Figure~\ref{fig:gran} for an example.

%%%%%%%%%%%%%%%%%%%%%%%%%%%%%%%%%%%%%%%%%%%%%%%%%%%%%%%%%%%
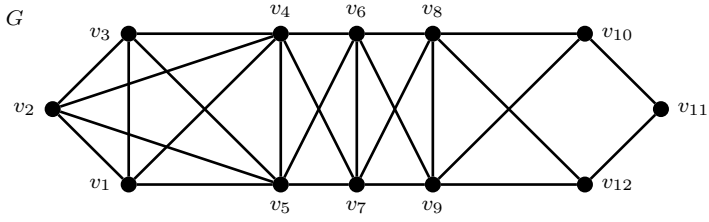
\begin{figure}[h]
	\begin{center}
\begin{tikzpicture}[scale=1]

\pgfsetlinewidth{1pt}

\tikzset{
	vertex/.style={circle,  draw, minimum size=5pt, inner sep=0pt}}

% grafo G
\draw (9.5,6) node[above] {$G$};

\node [vertex, fill] (v_2) at (10,5) [label=left:$v_2$]{};
\node [vertex, fill] (v_1) at (11,4) [label=left:$v_1$]{} edge (v_2);
\node [vertex, fill] (v_3) at (11,6) [label=left:$v_3$]{} edge (v_2) edge (v_1);

\node [vertex, fill] (v_{4}) at (13,6) [label=above:$v_{4}$]{} edge (v_1) edge (v_2) edge (v_3);
\node [vertex, fill] (v_{5}) at (13,4) [label=below:$v_{5}$]{} edge (v_1) edge (v_2) edge (v_3) edge (v_{4}) ;

\node [vertex, fill] (v_10) at (14,6) [label=above:$v_{6}$]{} edge (v_{5}) edge (v_{4});
\node [vertex, fill] (v_{7}) at (14,4) [label=below:$v_{7}$]{} edge (v_{4}) edge (v_10) edge (v_{5}) ;

\node [vertex, fill] (v_12) at (15,6) [label=above:$v_{8}$] {} edge (v_{7}) edge (v_10);
\node [vertex, fill] (v_11) at (15,4) [label=below:$v_{9}$]{} edge (v_{7}) edge (v_10) edge (v_12) ;

\node [vertex, fill] (v_{12}) at (17,4) [label=right:$v_{12}$]{} edge (v_11) edge (v_12);
\node [vertex, fill] (v_{11}) at (18,5) [label=right:$v_{11}$]{} edge (v_{12}); 
\node [vertex, fill] (v_{10}) at (17,6) [label=right:$v_{10}$]{} edge (v_11) edge (v_12) edge (v_{11});

\end{tikzpicture}
		
		\caption{The mp-subgraphs of $G$ are induced by the following sets $M_1 = \{v_1,v_2,v_3,v_4,v_5\}$, $M_2 = \{v_4,v_5,v_6,v_7\}$, $M_3 = \{v_6,v_7,v_8,v_9\},$ and $M_4 = \{v_8,v_9,v_{10},v_{11},v_{12}\}$. The extremal mp-subgraphs of $G$ are $M_1$ and $M_4$.}
		
		\label{fig:gran}
		
	\end{center}
\end{figure}
%%%%%%%%%%%%%%%%%%%%%%%%%%%%%%%%%%%%%%%%%%%%%%%%%%%%%%%%%%%

The following result on the monophonic convexity solves the problem of determining the toll hull number when the input graph is prime.

\begin{theorem} {\em \cite{DPS-2010-monophonic}} \label{thm:atom}
	If $G$ is a prime graph that is not a complete graph, then every pair of non-adjacent vertices is a monophonic hull set of $G$.
\end{theorem}

\begin{corollary} \label{cor:atom}
	Let $G$ be a prime graph. If $G$ is a complete graph, then the toll hull number of $G$ is $|V(G)|$; otherwise every two non-adjacent vertices form a toll hull set of $G$.
\end{corollary}

\begin{proof}
	If $G$ is a complete graph, it is clear that $V(G)$ is the only toll hull set of $G$.
	If $G$ is a not a complete graph, then the result follows from Theorem~\ref{thm:atom} because the monophonic convex hull of $S$ is contained in the toll convex hull of $S$ for any set $S \subseteq V(G)$.
	$\hfill\square$
\end{proof}

The following results state useful properties of reducible graphs.

\begin{lemma} {\em \cite{Leimer1993}} \label{lem:extremalmp}
	Every reducible graph has at least two extremal mp-subgraphs.
\end{lemma}

Given a tolled $(u,v)$-walk $W = w_1 \ldots w_k$, observe that if $uv \in E(G)$, then $W = uv$; and that if $w_2 = w_{k-1}$, then $W = uw_2v$.
A vertex $u$ of a t-convex set $S$ is {\em extreme in $S$} if $S - \{u\}$ is also a t-convex set.
Denote the set of toll extreme vertices of $V(G)$ by $Ext_t(G)$. It is clear that $Ext_t(G)$ is subset of every toll interval set and of every toll hull set of $G$ and that every toll extreme vertex is a simplicial vertex but the converse is not always true.

\begin{lemma} \label{lem:prime}
If $H$ is a subgraph of a graph $G$ such that every mp-subgraph of $H$ is also an mp-subgraph of $G$, then every induced $(u,v)$-path for $u,v \in V(H)$ contains at least one internal vertex of $H$ if $uv \not\in E(G)$.
\end{lemma}

\begin{proof}
If $u,v \in V(H)$, $uv \not\in E(G)$, and $P$ is an induced $(u,v)$-path whose none internal vertex belongs to $H$, then the union of $H$ and $P$ is a prime subgraph of $G$ properly containing an mp-subgraph of $G$, which is a contradiction.
$\hfill\square$
\end{proof}

We observe that $N(F)$ is a clique for every t-concave set $F$ that induces a connected graph, because if $u,v \in N(F)$ are non-adjacent, then any induced $(u,v)$-path of $G[F \cup \{u,v\}]$ is a tolled $(u,v)$-walk containing vertices of $F$.

\begin{lemma} {\em \cite{Alcon2015}} \label{lem:tollconvexset}
A vertex $v$ is in some tolled walk between non-adjacent vertices $x$ and $y$ if and only if $N[x] - \{v\}$ does not separate $v$ from $y$ and $N[y] - \{v\}$ does not separate $v$ from $x$.
\end{lemma}

Lemma~\ref{lem:tollconvexset} can be used to test whether a set $S$ is t-concave as follows.
In fact, we show how to compute the toll interval of $x$ and $y$ in polynomial time, which also allows one to test whether a vertex $v$ is toll extreme, which is equivalent to consider $S = \{v\}$, and to compute the toll convex hull of $S$ in polynomial time.
For every $x,y \in S$ and $v \not\in S$ such $xy \not\in E(G)$, if
$x$ and $v$ belong to the same connected component of $G - (N[y] - \{v\})$ and 
$y$ and $v$ belong to the same connected component of $G - (N[x] - \{v\})$, then
$v \in [x,y]$. This can be done in $O(n^3m)$ steps, where $n$ is the number of vertices and $m$ the number of edges.

The following result contains useful properties of tolled walks.

\begin{lemma} \label{lem:walk}
	Let $G$ be a graph, let $S \subset V(G)$, let $C \subseteq \semborda{S}$ be maximal such that $G[C]$ is connected, and let $x,y \not\in S$. The following sentences are equivalent.
	
	\begin{enumerate}[$(i)$]
		\item There is a tolled $(x,y)$-walk containing vertices of $C$. \label{ite:simple}
		
		\item There is a tolled $(x,y)$-walk containing vertices $x',y' \in \borda{S}$ such that $xy' \not\in E(G)$, $yx' \not\in E(G)$, $N(x') \cap C \neq \varnothing$, and $N(y') \cap C \neq \varnothing$. \label{ite:complex}
		
		\item $C \subset [x,y]$. \label{ite:Cxy}
	\end{enumerate}
\end{lemma}

\begin{proof}
	\bigskip \noindent $(\ref{ite:simple}) \Rightarrow (\ref{ite:complex})$ Let $W$ be a tolled $(x,y)$-walk containing vertex $v \in C$. Since $\borda{S}$ separates $\semborda{S}$ from $V(G) - S$ and $C$ is maximal contained in $\semborda{S}$ with $G[C]$ connected, we can write $W = x \ldots x' x'' \ldots v \ldots y'' y' \ldots y$ such that $x',y' \in \borda{S}$ and $x'',y'' \in C$. Therefore, $N(x') \cap C \neq \varnothing$, $N(y') \cap C \neq \varnothing$, and, by the definition of tolled walk, $xy' \not\in E(G)$ and $yx' \not\in E(G)$.
	
	\bigskip \noindent $(\ref{ite:complex}) \Rightarrow (\ref{ite:Cxy})$
	Let $W$ be a tolled $(x,y)$-walk containing vertices $x',y' \in \borda{S}$ such that $xy' \not\in E(G)$, $yx' \not\in E(G)$, $N(x') \cap C \neq \varnothing$, and $N(y') \cap C \neq \varnothing$.
	Let $x'' \in N(x') \cap C \neq \varnothing$ and 
	let $y'' \in N(y') \cap C \neq \varnothing$. Since $G[C]$ is connected, there is a $(x'',y'')$-walk $W'$ in $G[C]$ containing all vertices of $C$. (Note that $W'$ is not necessarily a tolled walk.) It is easy to see that
	the concatenation of $W_x, W',$ and $W_y$, where $W_x$ is a subwalk of $W$ from $x$ to $x'$ and $W_y$ is a subwalk of $W$ from $y'$ to $y$, is a tolled $(x,y)$-walk containing all vertices of $C$.
	
	\bigskip \noindent $(\ref{ite:Cxy}) \Rightarrow (\ref{ite:simple})$ Direct from the definitions.
	$\hfill\square$
\end{proof}

We conclude this section introducing the hull characteristic families.

If $F$ is a concave set of a convexity ${\cal C}$ on a set $X$, then every hull set of ${\cal C}$ contains at least one vertex of $F$. We define the {\em granularity of $F$ under ${\cal C}$} as the maximum integer $g(F)$ such that every hull set of ${\cal C}$ has at least $g(F)$ vertices of $F$.
Let ${\cal F}$ be a family of pairwise disjoint concave sets of ${\cal C}$. The {\em granularity of ${\cal F}$} is the sum of the granularities of its members.
We say that ${\cal F}$ is a {\em hull characteristic family of ${\cal C}$} if the hull number of ${\cal C}$ is equal to the granularity of ${\cal F}$.

The problem of computing the hull number of ${\cal C}$ can be reduced to the one of finding a hull characteristic family of ${\cal C}$ and computing the granularity of each of its members. The family formed only by $X$ is itself a trivial hull characteristic family of ${\cal C}$, but it brings no advantage of the use of this notion for determining the hull number of ${\cal C}$.
The number of hull characteristic families of ${\cal C}$ can be an exponential on the size of $X$. For instance, every partition of the vertex set $V(G)$, where $G$ is a complete graph, is a hull characteristic family of the toll convexity of $G$, since the toll hull number of $G$ is $|V(G)|$ when $G$ is a complete graph.
An example of a non-trivial hull characteristic family in toll convexity is the family ${\cal C} = \{S_1 = \{v_1,v_2,v_3\}, S_2 = \{v_{10},v_{11},v_{12}\}\}$ of vertices of the graph $G$ of Figure~\ref{fig:gran}. One can use Lemma~\ref{lem:tollconvexset} to see that the members of $\cal C$ are really t-concave sets. In fact, this lemma can be used to show that all vertices of $S_1$ are extreme vertices, then $g(S_1) = 3$. Since $S_1$ is not a toll hull set of $G$, the toll hull number of $G$ is at least 4. Now, considering the tolled walks $v_3v_4v_6v_8v_{10}v_{11}$ and $v_1v_5v_7v_9v_{12}v_{11}$, we conclude that $S_1 \cup \{v_{11}\}$ is a toll hull set of $G$, that $g(S_2) = 1$, and also that the toll hull number of $G$ is $4$.

%%%%%%%%%%%%%%%%%%%%%%%%%%%%%%%%%%%%%%%%%%%%%%%%%%%%%%%%%%%%%%%%%%%%%%%%%%%%%%%%%%%%
\section{The algorithm} \label{sec:hullnumber}

The central idea of the proposed algorithm is to find a toll hull characteristic family $\cal C$ of the input graph such that the granularity of each member of $\cal C$ can be determined in polynomial time. In order to get this, the algorithm begins constructing families ${\cal F}$ and ${\cal M}$ of vertex sets such that $\semborda{F} \cap \semborda{F'} = \varnothing$ for any $F, F' \in {\cal F} \cup {\cal M}$.
During the algorithm, the members of $F \in {\cal F}$ such that $\semborda{F}$ is not t-concave can be joined with other members of ${\cal F} \cup {\cal M}$ so that, at the end, the sets $\semborda{F}$ that are t-concave sets form the desired family. The following classification of the t-concave sets $F$ of a graph is useful to accomplish this task.

$$\mbox{Type of } F =
\left\{
\begin{array}{rl}
1 ,& \mbox{ if } (F,N(F)) \mbox{ is not complete} \\
2 ,& \mbox{ if } (F,N(F)) \mbox{ is complete and } F \mbox{ is not a clique} \\
3 ,& \mbox{ if } F \cup N(F) \mbox{ is a clique} \\
\end{array}
\right.
$$

\begin{lemma} \label{lem:granularity}
If $F$ is a t-concave set of a graph $G$, then $g(F) \geq i$ if the type of $F$ is $i \in \{1,2\}$ and $g(F) = |F|$ if the type of $F$ is $3$.
\end{lemma}

\begin{proof}
The case $t = 1$ is trivial. For $t = 2$, suppose for contradiction that $S$ is a toll hull set of $G$ such that $S \cap F = \{x\}$. The type of $F$ implies $|F| \ge 2$. Since $F$ is t-concave and $F - \{x\} \subset \langle S \rangle$ there is, for some $y \not\in F$, a tolled $(y,x)$-walk containing some vertex $v \in F - \{x\}$. However, since $N(F) \subset N[x]$, $N[x] - \{v\}$ separates $v$ from $y$, which contradicts Lemma~\ref{lem:tollconvexset}.

Finally consider $t = 3$. We claim that all vertices of $F$ are toll extreme vertices, which implies that $F \subseteq S$. Suppose the contrary and let $W$ be a tolled $(x,y)$-walk containing some vertex $v \in F - \{x,y\}$. Since $F$ is t-concave, at least one, say $x$, belongs to $F$. Since $F \cup N(F)$ is a clique, $y \not\in F \cup N(F)$. Now, the fact that $N[x] - \{v\}$ separates $v$ from $y$ contradicts Lemma~\ref{lem:tollconvexset}.
$\hfill\square$	
\end{proof}

An example of a t-concave set with granularity strictly bigger than its type is the set $F = \{v_1,v_2,v_3,v_4,v_5,v_6,v_7\}$ of Figure~\ref{fig:gran}, since the type of $F$ is $1$ and $g(F) \geq 3$ because vertices $v_1,v_2,v_3$ are toll extreme vertices of the graph. 

Once a t-concave set $F$ is formed by the algorithm, it can be chosen in a later iteration to compose another t-concave set. The vertices of $F$ that will be chosen to constitute the minimum toll hull set returned by the algorithm depends of the type of $F$ and of other properties that $F$ has at the moment that $F$ is formed. They are detailed in the following numbered choices.
The first 3 choices are for t-concave sets of type~1.

\begin{choice} \label{cho:T1}
add to $S$ a vertex $u \in \semborda{F^\bullet}$ having a non-neighbor in $\borda{F^\bullet}$	
\end{choice}

\begin{choice} \label{cho:T1F0}
add to $S$ a vertex $u \in \semborda{F^\bullet}$ having a non-neighbor $u'$ in $\borda{F^\bullet}$ and a non-neighbor $u''$ in $\borda{F^\circ}$, and there are different $F_1, F_2 \in {\cal M'} \cup {\cal F'}$ such that $u \in \semborda{F_1}$ and $\borda{F^\bullet} \subset F_2$.
\end{choice}

\begin{choice} \label{cho:T1FJ}
add to $S$ a vertex $u \in \semborda{F^\bullet}$ having a non-neighbor in $\borda{F^\bullet}$ if there are different $F_1, F_2 \in {\cal M'} \cup {\cal F'}$ such that $u \in \semborda{F_1}$ and $\borda{F^\bullet} \subset F_2$.
\end{choice}

The remaining 5 choices are for t-concave sets of type 2.

\begin{choice} \label{cho:T2}
add to $S$ non-adjacent vertices $u_1, u_2 \in \semborda{F^\bullet}$.
\end{choice}

\begin{choice} \label{cho:T20T1a}
add to $S$ non-adjacent vertices $u_1, u_2 \in \semborda{F_1}$ for some $F_1 \in {\cal M'} \cup {\cal F'}$ both having non-neighbors $u'_1, u'_2 \in \borda{F^\circ}$, respectively, such that $N[u_1]$ does not separate $u_2$ from $u'_1$, and $N[u_2]$ does not separate $u_1$ from $u'_2$.
\end{choice}

\begin{choice} \label{cho:T20T1b}
for $i \in \{1,2\}$, add to $S$ a vertex $u_i \in \semborda{F_i}$ for some $F_i \in {\cal M'} \cup {\cal F'}$ having a non-neighbor $u'_i \in \borda{F^\circ}$ such that $F_1 \ne F_2$, $N[u_1]$ does not separate $u_2$ from $u'_1$, and $N[u_2]$ does not separate $u_1$ from $u'_2$.
\end{choice}

\begin{choice} \label{cho:T21T1}
add to $S$ a vertex $u_2 \in \semborda{F_2}$ for some $F_2 \in {\cal M'} \cup ({\cal F'} - \{F_1\})$ having a non-neighbor in $\borda{F^\circ}$ where $\semborda{F_1}$ is t-concave of type~$1$.
\end{choice}

\begin{choice} \label{cho:T21T2}
add to $S$ a vertex $u_2 \in \semborda{F_2}$ for some $F_2 \in {\cal M'} \cup {\cal F'} - \{F_1\}$ where $\semborda{F_1}$ is t-concave of type~$1$.
\end{choice}

%%%%%%%%%%%%%%%%%%%%%%%%%%%%%%%%%%%%%%%%%%%%%%%%%%%%%%%%%%%
\begin{algorithm}[h!]
	
	\SetKwInOut{Input}{input}
	\SetKwInOut{Output}{output}
	
	\Input{A graph $G$}
	\Output{A minimum toll hull set of $G$}
	
	\If{$V(G)$ is a clique}{
		\Return $V(G)$
	}
	\If{$G$ is prime}{ \label{lin:prime}
		\Return two non-adjacent vertices of $G$
	}
	
	compute the mp-subgraphs of $G$ \label{lin:decomposition}
	
	${\cal F} \leftarrow \{F : F$ is the vertex set of an extremal mp-subgraph of $G \}$ \label{lin:extremal}

	${\cal M} \leftarrow \{F : F$ is the vertex set of a non-extremal mp-subgraph of $G$$\}$ \label{lin:othermp}
	
	$S \leftarrow \varnothing$
	
	\For{$F \in {\cal F}$}{ \label{lin:mp-subgraph-beg}
		\If{$\semborda{F}$ is t-concave of type $1$}{apply Choice~\ref{cho:T1} \label{lin:choT1}}
		\If{$\semborda{F}$ is t-concave of type $2$}{apply Choice~\ref{cho:T2} \label{lin:choT2}}
		\If{$\semborda{F}$ is t-concave of type $3$}{add $\semborda{F}$ to $S$} \label{lin:mp-subgraph-end}
	} 

	\While{there is $F^\circ \in {\cal F}$ such $\semborda{F^\circ}$ is not t-concave and $\borda{F^\circ} \subset F$ for some $F \in {\cal M} \cup {\cal F} - \{F^\circ\}$} { \label{lin:F}
		
		${\cal M'} \leftarrow \{ F : F \in {\cal M}$ and $\borda{F^\circ} \subseteq F\}$ \label{lin:H}
		
		${\cal M} \leftarrow {\cal M} - {\cal M'}$ \label{lin:removeM}
		
		${\cal F'} \leftarrow \{ F : F \in {\cal F}$ and $\borda{F^\circ} \subseteq F\}$ \label{lin:H2}
		
		$F^\bullet \leftarrow \underset{F \in {\cal M'} \cup {\cal F'}}{\bigcup} F$ \label{lin:F3}
		
		${\cal F} \leftarrow ({\cal F} - {\cal F'}) \cup \{F^\bullet\}$ \label{lin:join}
		
		\If{$\semborda{F^\bullet}$ is t-concave of type $i$}{ \label{lin:concave}
			
			let $k$ be the number of members $F$ of ${\cal F'}$ such that $\semborda{F}$ is t-concave \label{lin:F*beg}
			
			\If{$i = 1$ and $k = 0$}{if possible, apply Choice~\ref{cho:T1F0}; else apply Choice~\ref{cho:T1FJ} \label{lin:choT1k0}}
			
			\If{$i = 2$ and $k = 0$}{if possible, apply Choice~\ref{cho:T20T1a}; else apply Choice~\ref{cho:T20T1b} \label{lin:choT2k0}}
			
			\If{$i = 2$ and $k = 1$}{if possible, apply Choice~\ref{cho:T21T1}; else apply Choice~\ref{cho:T21T2} \label{lin:choT2k1}} \label{lin:F*end}
		}
	}
	\Return $S$ \label{lin:return}
	
	\caption{Minimum toll hull set \label{alg:thn}}
	
\end{algorithm}
%%%%%%%%%%%%%%%%%%%%%%%%%%%%%%%%%%%%%%%%%%%%%%%%%%%%%%%%%%%

\begin{lemma} \label{lem:F}
The following sentences hold for $F$ and $F'$ chosen at the same time from ${\cal F} \cup {\cal M}$ in Algorithm~$\ref{alg:thn}$.

\begin{enumerate}[$(i)$]
	\item If $F \in {\cal F}$, then $\semborda{F} \ne \varnothing$. \label{ite:non-empty}
	\item $\semborda{F} \cap \semborda{F'} = \varnothing$. \label{ite:disjoint}	
	\item $G[F - C]$ is connected for every clique $C \subseteq F'$. \label{ite:maximal}
	\item $F \cap F'$ is a clique. \label{ite:FM}
	\item If $\borda{F} \not\subseteq F''$ for every $F'' \in ({\cal F} \cup {\cal M}) - F$, then $G-F$ is disconnected. \label{ite:border-not-contained}
\end{enumerate}
\end{lemma}

\begin{proof}
\bigskip\noindent$(\ref{ite:non-empty})$ and~$(\ref{ite:disjoint})$
The mp-subgraphs of $G$ are obtained in line~\ref{lin:decomposition}. The extremal mp-subgraphs of $G$ form family ${\cal F}$ (line~\ref{lin:extremal}) and the remaining ones form family ${\cal M}$ (line~\ref{lin:othermp}).
Note that for any mp-subgraph $M$ of $G$, $M$ is the only mp-subgraph containing the vertices of $\semborda{M}$.
Since the definition of extremal mp-subgraph implies that $\semborda{M} \ne \varnothing$ for every extremal mp-subgraph $M$ of $G$, $(\ref{ite:non-empty})$ and~$(\ref{ite:disjoint})$ hold when line~\ref{lin:othermp} finishes.

The only operations performed on ${\cal M}$ are deletions and on ${\cal F}$ are deletions and inclusions. Since deletions do not interfere in the properties of items~$(\ref{ite:non-empty})$ and~$(\ref{ite:disjoint})$, we only need to consider the inclusions on ${\cal F}$.
The only line adding members to ${\cal F}$ after line~\ref{lin:extremal} is line~\ref{lin:join}, which is inside of the {\bfseries while} loop.

A set $F^\bullet$ constructed in line~\ref{lin:F3} contains the set $F^\circ \in {\cal F}$ chosen in line~\ref{lin:F} of the same iteration. Therefore, every member of ${\cal F}$ at the end of iteration $k-1$ contains the vertex set of at least one extremal mp-subgraph of $G$, which means that $\semborda{F^\bullet}$ is non-empty and $(\ref{ite:non-empty})$ does hold.

Now, note that every member belonging to ${\cal F} \cup {\cal M}$ at the end of iteration $k \ge 1$ of the {\bfseries while} loop is present in ${\cal F} \cup {\cal M}$ at the beginning of iteration $k$, except the new one, namely, the set $F^\bullet$ constructed in line~\ref{lin:F3}. Therefore, if there are two sets $F, F' \in {\cal F} \cup {\cal M}$ such that $\semborda{F} \cap \semborda{F'} \ne \varnothing$ at the end of iteration $k$ and $k$ is minimum with this property, them one of them is $F^\bullet$, say $F = F^\bullet$. However, observe that $F'$ and $\borda{F'}$ do not change from the beginning to the end of iteration $k$, which implies that $\semborda{F^\bullet} \cap \semborda{F'} = \varnothing$, a contradiction. Then $(\ref{ite:disjoint})$ does hold.

\bigskip\noindent$(\ref{ite:maximal})$ Suppose that $G[F - C]$ is disconnected.
If $F$ is an mp-subgraph of $G$, then it is clear that $G[F - C]$ is connected.
Therefore, we can choose $F$ as the first set added to ${\cal F}$ in line~\ref{lin:join} such that $G[F - C]$ is disconnected. Write $C' = C \cap F$. Note that $C'$ is not equal to any mp-subgraph of $G$ forming $F$ and that there is an mp-subgraph $M'$ of $G$ composing $F'$ which contains $C$. Furthermore, since $M - C'$ is connected for every mp-subgraph $M$ of $G$, there are mp-subgraphs $M_1$ and $M_2$ of $G$ doing part of $F$ such that $V(M_1) \cap V(M_2) \subseteq C'$ and $V(M_1) - C'$ belongs to connected component of $G[F - C]$ different of the one containing $V(M_2) - C'$. Hence, the set $F^\circ$, chosen in line~$\ref{lin:F}$ of some iteration $k$ of the {\bfseries while} loop, has the property that $\borda{F^\circ} \subseteq V(M_1) \cap V(M_2)$. Therefore, the member of ${\cal F} \cup {\cal M}$ containing $M'$ should have been chosen to compose the set $F^\bullet$, constructed in line~\ref{lin:F3} of iteration $k$, which is a subset of $F$, yielding a contradiction because every mp-subgraph of $G$ belongs to exactly one member of ${\cal F} \cup {\cal M}$ during all the algorithm.

\bigskip\noindent$(\ref{ite:FM})$ Since the intersection of two mp-subgraphs is a clique, we can consider that $k$ is the number of the least iteration of the {\bfseries while} loop that finishes with ${\cal F} \cup {\cal M}$ having members $F$ and $F'$ whose intersection is not a clique. Without loss of generality, we can assume that $F = F^\bullet$, where $F^\bullet$ is the set added to ${\cal F}$ in line~\ref{lin:join} of iteration $k$. Observe that the intersection of $F'$ and every member of ${\cal F'} \cup {\cal M'}$ of iteration $k$ is a clique. Therefore, $F'$ intersects with at least two members $F_1, F_2 \in {\cal F'} \cup {\cal M'}$ and there are $w_1 \in (F' \cap F_1) - F_2$ and $w_2 \in (F' \cap F_2) - F_1$. We know that there is a vertex $v \in \borda{F^\circ} - F'$ because otherwise $F'$ would have been chosen to be part of $F^\bullet$.
Let $P_1$ be an induced $(w_1,v)$-path of $G[F_1 - F']$ and
let $P_2$ be an induced $(w_2,v)$-path of $G[F_2 - F']$. These paths exist because~$(\ref{ite:maximal})$ guarantee that $G[F_1 - F']$ and $G[F_2 - F']$ are connected graphs once we know that the intersection of $F'$ with any member of f ${\cal F'} \cup {\cal M'}$ of iteration $k$ is a clique. Now, from the concatenation of $P_1$ and $P_2$ we can obtain an induced $(w_1,w_2)$-path whose internal vertices do not belong to $F'$, which contradicts the hypothesis that $F'$ is composed by mp-subgraphs.

\bigskip\noindent$(\ref{ite:border-not-contained})$
Let $F$ be such that $\borda{F}$ is not contained in any other member of ${\cal F} \cup {\cal M}$ and let $M_1$ be an mp-subgraph of $G$ not composing $F$ such that $C_1 = F \cap M_1$ is maximal. Therefore, there is an mp-subgraph $M_2 \ne M_1$ of $G$ not composing $F$ such that $C_2 = F \cap M_2$ satisfies $v_2 \in C_2 - C_1$. By the maximality of $C_1$, there is $v_1 \in C_1 - C_2$. Choose $u_i \in M_i - C_i$ for $i \in \{1,2\}$. Vertex $u_i$ exists because $M_i$ is an mp-subgraph for $i \in \{1,2\}$.

First, we show that $C_1 \cup C_2$ is a clique. Suppose the contrary. By~$(\ref{ite:FM})$, there are $u_1 \in C_1 - C_2$ and $u_2 \in C_2 - C_1$ such that $u_1u_2 \not\in E(G)$. Therefore, if $G - F$ is connected, since both $u_1$ and $u_2$ have neighbors in $G - F$, there is an induced $(u_1,u_2)$-path whose internal vertices do not belong to $F$, which contradicts Lemma~\ref{lem:prime}.

Next, we show that $V(M_1) \cap V(M_2) \subset F$. Suppose the contrary and let $w \in V(M_1) \cap V(M_2) - F$. By the choice of $M_1$ and $M_2$ and the fact that $C_1 \cup C_2$ is a clique, $w$ has a non-neighbor $w' \in C_1 \cup C_2$. Assume that $w' \in C_1$. (The case where $w' \in C_2$ is analogue.) Therefore, using $F$ and $M_2$, we can find an induced $(w,w')$-path whose internal vertices do not belong to $M_1$, which contradicts the assumption that $M_1$ is an mp-subgraph once $ww' \not\in E(G)$.

Now, supposing that $G - F$ is connected and let $P$ be a $(u_1,u_2)$-path of $G - F$, we reach in a contradiction because the union of $M_1$, $M_2$, and $P$ is a prime subgraph of $G$ properly containing an mp-subgraph.
$\hfill\square$
\end{proof}

\begin{corollary} \label{cor:non-extremal}
	If $M$ is a non-extremal mp-subgraph of $G$, then $G - M$ is disconnected.
\end{corollary}

\begin{proof}
It is a special case of Lemma~\ref{lem:F}~$(\ref{ite:border-not-contained})$ when $F$ is a non-extremal mp-subgraph. $\hfill\square$
\end{proof}

\begin{lemma} \label{lem:concave}
The following sentences hold at any time of Algorithm~$\ref{alg:thn}$ for $F \in {\cal F}$ if $\semborda{F}$ is t-concave.
	
\begin{enumerate}[$(1)$]
	\item $\borda{F}$ is contained in an mp-subgraph not composing $F$. \label{ite:inter}
	\item $G[\semborda{F}]$ is connected. \label{ite:connected}
	\item $\borda{F}$ is a clique. \label{ite:clique}
	\item $\borda{F} = N(\semborda{F})$. \label{ite:borda}
\end{enumerate}
	
\end{lemma}

\begin{proof}
\bigskip\noindent$(\ref{ite:inter})$ Suppose by contradiction that $\borda{F}$ is not contained in a mp-subgraph of $G$ not composing $F$. Then, there are mp-subgraphs $M$ and $M'$ of $G$ not composing $F$ such that there are $v \in (\borda{F} \cap M) - M'$ and $v' \in (\borda{F} \cap M') - M$.
By the definition of mp-subgraph, we can choose $u \in V(M) - V(M')$ and $u' \in V(M') - V(M)$ such that $uu' \not\in E(G)$.
Let $P_u$ be an induced $(u,v)$-path of $M - V(M')$ and
let $P_{u'}$ be an induced $(u',v')$-path of $M' - V(M)$.
Note that $C = (N(u) \cap F) - \{v\}$ and $C' = (N(u') \cap F) - \{v'\}$ are cliques.
By Lemma~\ref{lem:F}~$(\ref{ite:maximal})$, $G[F - C]$ and $G[F - C']$ are connected subgraphs of $G$.
Lemma~\ref{lem:F}~$(\ref{ite:non-empty})$ guarantees that there is a vertex $w \in \semborda{F}$.
Since $w$ belongs to $G[F - C]$ and to $G[F - C']$, 
there is an induced $(v,w)$-path $P_w$ in $G[F - C]$ and 
there is an induced $(v',w)$-path $P_{w'}$ in $G[F - C']$.
Since the concatenation of the paths $P_u, P_w, P_{w'}$, and $P_{u'}$ is a tolled $(u,u')$-walk containing $w$, we have a contradiction.

\bigskip\noindent$(\ref{ite:connected})$ Suppose by contradiction that $G[\semborda{F}]$ is disconnected. If $F$ was added to ${\cal F}$ in line~\ref{lin:extremal}, then $G[\semborda{F}]$ is connected because $F$ is an extremal mp-subgraph of $G$. Then, we can choose $F$ as the first set added to ${\cal F}$ in line~\ref{lin:join} such that $\semborda{F}$ is t-concave and $G[\semborda{F}]$ is disconnected. Therefore, $F$ is equal to a set $F^\bullet$ constructed in line~\ref{lin:F3} of some iteration $k$. Let $H_1$ and $H_2$ be connected components of $G[\semborda{F}]$. Hence, $H_1$ and $H_2$ belong to different members $F_1$ and $F_2$ of ${\cal F'} \cup {\cal M'}$ constructed in iteration $k$. Therefore, for the set $F^\circ$ chosen in line~\ref{lin:F} of iteration $k$, it holds that $\borda{F^\circ}$ is contained in $F_1$ and in $F_2$.
On the one hand, $\borda{F}$ contains $\borda{F^\circ}$. Since item~$(\ref{ite:inter})$ guarantees that there is an mp-subgraph $M$ not composing $F$ containing $\borda{F}$, there is a member of ${\cal F} \cup {\cal M}$ that should have been chosen to compose $\borda{F}$ in iteration $k$ that was not chosen, which is a contradiction.
On the other hand, there is a vertex of $\borda{F^\circ}$ present in $H_1$ and in $H_2$, but this is a contradiction because different connected components have disjoint vertex sets.

\bigskip\noindent$(\ref{ite:clique})$ It is a consequence of item~$(\ref{ite:inter})$ and Lemma~\ref{lem:F}~$(\ref{ite:FM})$. 

\bigskip\noindent$(\ref{ite:borda})$ By the definitions, it holds $N(\semborda{S}) \subseteq \borda{S}$ for any set $S$.
Now, if $F \in {\cal F}$ and $\semborda{F}$ is t-concave, by item~$(\ref{ite:inter})$, there are mp-subgraphs $M$ and $M'$ properly containing $\borda{F}$ such that $V(M) \subseteq F$ and $V(M') \not\subseteq F$. If some vertex $v \in \borda{F}$ had no neighbor in $\semborda{F}$, then $\borda{F} - \{v\}$ would be a clique separator of $M$. Then $\borda{F} \subseteq N(\semborda{F})$.
$\hfill\square$
\end{proof}

Whenever Algorithm~\ref{alg:thn} constructs a t-concave set $F$, some vertices of $F$ are choosen appropriatly to form the toll hull set that will be returned at the end. We show in the next result that at least one choice of each line is always possible.

\begin{lemma} \label{lem:choice}
	In Algorithm~$\ref{alg:thn}$, Choices~$\ref{cho:T1},~\ref{cho:T2},~\ref{cho:T1FJ},$ and $~\ref{cho:T21T2}$ are always possible in lines~$\ref{lin:choT1},\ref{lin:choT2},\ref{lin:choT1k0},$ and $\ref{lin:choT2k1}$, respectively, and Choice$~\ref{cho:T20T1a}$ or$~\ref{cho:T20T1b}$ is possible in line~$\ref{lin:choT2k0}$.
\end{lemma}

\begin{proof}
	Since it is easy to see that
	Choice~$\ref{cho:T1}$ is always possible for a t-concave set having type~1 and
	Choices~$\ref{cho:T2}$ and~\ref{cho:T21T2} are always possible for a t-concave set having type~2, we discuss the choices of lines~\ref{lin:choT1k0} and~\ref{lin:choT2k0} separatedly in the sequel.
	
	Now, we show that Choice~\ref{cho:T1FJ} is possible in line~\ref{lin:choT1k0}. By definition of type~1, there is $u \in \semborda{F^\bullet}$ with a non-neighbor in $\borda{F^\bullet}$. On the one hand $\borda{F^\bullet} \subset \borda{F^\circ}$. By Lemma~\ref{lem:F}~$(\ref{ite:FM})$, $\borda{F^\circ}$ is a clique. Let $F_1 \in {\cal M'} \cup {\cal F'}$ be a set containing $u$. If $u \in \borda{F_1}$, then there is $F' \in {\cal F'} \cup {\cal M'} - F_1$ such that $u \in F'$. Using Lemma~\ref{lem:F}~$(\ref{ite:FM})$ again, we have that $F_1 \cap F'$ is a clique containing $\borda{F^\bullet}$, which would imply that $u$ is adjacent to all vertices of $\borda{F^\bullet}$. Then, $u \in \semborda{F_1}$. Choose any other member of ${\cal M'} \cup {\cal F'}$ as $F_2$ and Choice~\ref{cho:T1FJ} is possible for this case.
	On the other hand $\borda{F^\bullet} \not\subseteq \borda{F^\circ}$. Then, there is one member of ${\cal M'} \cup {\cal F'}$ not containing $\borda{F^\bullet}$, call it $F_1$. Any vertex of $\semborda{F_1}$ can be chosen as $u$.
	Choose as $F_2$ any member containing $\borda{F^\bullet}$. It there exists because $\borda{F^\bullet}$ is a clique. Then, Choice~\ref{cho:T1FJ} is possible in line~\ref{lin:choT1k0}.
	
	Finally, we show that Choice~\ref{cho:T20T1a} or~\ref{cho:T20T1b} is possible in line~\ref{lin:choT2k0}. Since $\semborda{F^\circ}$ is not t-concave and $\borda{F^\circ}$ is a clique, there are vertices $u_1,u_2 \not\in F^\circ$ for which there is a tolled $(u_1,u_2)$-walk containing vertices of $\semborda{F^\circ}$. Since $\semborda{F^\bullet}$ is t-concave of type 2, $(\borda{F^\bullet},\semborda{F^\bullet})$ is complete, which implies that both $u_1,u_2$ belong to $\semborda{F^\bullet}$. If there is $F_1$ such that $u_1,u_2 \in F_1 - (\borda{F^\bullet} \cup \borda{F^\circ})$, we have Choice~\ref{cho:T20T1a} because if $u_i \in \borda{F_i}$ for some $i \in \{1,2\}$, then $u_i$ is adjacent to all vertices of $\borda{F^\circ}$ or $u_i \in \borda{F^\bullet}$.
	The other possibility is $u_1 \in F_1 - (\borda{F^\bullet} \cup \borda{F^\circ})$ and $u_2 \in F_2 - (\borda{F^\bullet} \cup \borda{F^\circ})$ for $F_1,F_2 \in {\cal M'} \cup {\cal F'} - \{F^\circ\}$, which matches with Choice~\ref{cho:T20T1b}.
	$\hfill\square$
\end{proof}

The following result guarantees that if the $F^\bullet$ constructed in line~\ref{lin:F3} of Algorithm~\ref{alg:thn} is such that $\semborda{F^\bullet}$ is t-concave, then at most two sets $F$ used to compose $F^\bullet$ are such that $\semborda{F}$ is t-concave. Furthermore, the type of $\semborda{F}$ is 1.

\begin{lemma} \label{lem:formed}
	Let ${\cal F'}$ and $F^\bullet$ be obtained in lines~$\ref{lin:H2}$ and~$\ref{lin:F3}$ of some iteration $k$ of the {\bfseries while} loop of Algorithm~$\ref{alg:thn}$, respectively. If $\semborda{F^\bullet}$ is t-concave, then ${\cal F'}$ has at most two sets $F$ such that $\semborda{F}$ is t-concave and the type of $\semborda{F}$ is~$1$.
\end{lemma}

\begin{proof}
Let $F^\circ$ be the set chosen in line~$\ref{lin:F}$ of iteration $k$ and let $\{u,v\} \not\in \semborda{F^\circ}$ be such that there is a tolled $(u,v)$-walk $W$ containing some vertex of $\semborda{F^\circ}$.
Denote by ${\cal F''}$ the family formed by the sets $F \in {\cal F'}$ such that $\semborda{F}$ is t-concave.
First, suppose by contradiction that $F_1, F_2, F_3 \in {\cal F''}$. Since $\semborda{F_1}$ is t-concave, Lemma~\ref{lem:concave}~$(\ref{ite:clique})$ implies that $\borda{F_1}$ is a clique. Since $\borda{F^\circ} \subseteq \borda{F_1}$, we conclude that $\borda{F^\circ}$ is a clique, which means that that $u,v \not\in F^\circ$.

By Lemma~\ref{lem:F}~$(\ref{ite:disjoint})$, the sets $\semborda{F_1},\semborda{F_2},$ and $\semborda{F_3}$ are pairwise disjoint. Hence, for at least one of them, say $\semborda{F_3}$, it holds $u,v \not\in \semborda{F_3}$. Furthermore, we have that $u,v \not\in \borda{F_i}$ for $i \in \{1,2,3\}$ because each of $u$ and $v$ has at least one non-neighbour in $\borda{F^\circ}$, $\borda{F^\circ} \subseteq \borda{F_i}$ for $i \in \{1,2,3\}$, and $\borda{F_i}$ is a clique. 
Note that $W$ has different vertices $u', v' \in \borda{F^\circ}$ such that $uv' \not\in E(G)$ and $vu' \not\in E(G)$.
By Lemma~\ref{lem:concave}~$(\ref{ite:borda})$, every vertex of $\borda{F_3}$ has at least one neighbor in $\semborda{F_3}$.
Since $\borda{F^\circ} \subseteq \borda{F_3}$ and $G[\semborda{F_3}]$ is connected by Lemma~\ref{lem:concave}~$(\ref{ite:connected})$, 
Lemma~\ref{lem:walk}~$(\ref{ite:complex})$ implies that $\semborda{F_3}$ $\subset [u,v]$, which contradicts the assumption that $\semborda{F_3}$ is t-concave. Therefore, $|{\cal F''}| \le 2$.

Now, suppose by contradiction that $\semborda{F_1}$ is a t-concave set of type 2 or 3. Hence $(\borda{F_1},\semborda{F_1})$ is complete, which implies that $u,v \not\in F_1$ because each of $u$ and $v$ has at least one non-neighbor in $\borda{F^\circ}$, $\borda{F^\circ} \subseteq \borda{F_1}$, and $\borda{F_1}$ is a clique. As in the previous case, we have that $\semborda{F_1} \subset [u,v]$, which is a contradiction because $\semborda{F_1}$ is a t-concave set.
$\hfill\square$
\end{proof}

From now on, denote by 
$S^*, {\cal F^*},$ and ${\cal M^*}$ the instances of the set $S$ and the families ${\cal F}$ and ${\cal M^*}$ of Algorithm~\ref{alg:thn} when line~\ref{lin:return} is reached, respectively, and write ${\cal C} = \{\semborda{F} : F \in {\cal F^*}$ and $\semborda{F}$ is t-concave$\}$. We will show that ${\cal C}$ is a toll hull characteristic family of the input graph.

\begin{lemma} \label{lem:complement}
Let $F \in {\cal F^*}$. For every connected component $H$ of $G - F$, there is $F' \in {\cal F^*}$ such that $\semborda{F'} \subseteq V(H)$ and $\semborda{F'}$ is t-concave.
\end{lemma}

\begin{proof}
Let $F$ be the set of ${\cal F^*}$ minimizing the number of vertices of a connected component of $G - F$ and let $H$ be a connected component of $G - F$ with minimum number of vertices.
Since $F$ is the union of mp-subgraphs of $G$, every mp-subgraph of $G' = G[V(H) \cup F]$ is also an mp-subgraph of $G$. Hence $G'$ is reducible, which means that $G'$ has at least two extremal mp-subgraphs such that at least one of them is contained in $V(H) \cup N(H)$. Call it by $M$.
Therefore, $M$ is contained in a member $F_M$ of ${\cal F^*}$ contained in $V(H) \cup N(H)$. If $F_M$ is not t-concave, then $\borda{F_M}$ is not contained in one member of ${\cal F^*} \cup {\cal M^*}$. By Lemma~\ref{lem:F}~$(\ref{ite:border-not-contained})$, $G - F_M$ is disconnectec containing a connectec component with less vertices than $H$, which is a contradiction. Therefore, $F_M \in {\cal F^*}$ is such that $\semborda{F_M} \subseteq V(H)$ and $\semborda{F_M}$ is t-concave.
$\hfill\square$
\end{proof}

The next result is essential to show that only one vertex suffices for every set $C \in {\cal C}$ having type~1.

\begin{lemma} \label{lem:1T1}
If $F \in {\cal F^*}$ is such that $\semborda{F}$ is t-concave of type~$1$, then $S^* \cap \semborda{F} = \{u\}$ and $\semborda{F} \subset \langle u,w \rangle$ for any $w \not\in F$.
\end{lemma}

\begin{proof}
Let $F^\bullet \in {\cal F^*}$ be such that $\semborda{F^\bullet}$ is t-concave of type~1. By Lemma~\ref{lem:choice}, there are $u \in S^* \cap \semborda{F^\bullet}$ and $u' \in \borda{F^\bullet}$ such that $u' \not\in N(u)$. Let $H$ be a connected component of $G - F^\bullet$ such that there is $u'' \in V(H)$ with $u' \in N(u'')$. By Lemma~\ref{lem:complement}, there is $w \in S^* \cap V(H)$. Let $Q$ be a $(u,u''')$-path of $G[\semborda{F^\bullet}]$ where $u''' \in N(u')$ and let $Q'$ be a $(w,u'')$-path of $H$.
Write $P$ as the concatenation of $Q$ with the edge $u'''u'$, and
$P_{wu'}$ as the concatenation of $Q'$ with the edge $u'u''$.
It is clear that the concatenation of $Q$ with $Q'$ is a tolled $(u,w)$-walk containing $u'$.

If $F^\bullet$ was added to ${\cal F}$ in line~\ref{lin:extremal}, then, $F^\bullet$ is an extremal mp-subgraph of $G$. 
By Theorem~\ref{thm:atom}, $F^\bullet \subset \langle u,u' \rangle$, which means that $\semborda{F^\bullet} \subset \langle u,w \rangle$.
The other possibility is that $F^\bullet$ was added to $\cal F$ in line~\ref{lin:join} of some iteration $k$ of the {\bfseries while} loop. Let $\semborda{F^\circ}$, ${\cal F'}$, and ${\cal M'}$ be obtained in lines~\ref{lin:F}, ~\ref{lin:H}, and~\ref{lin:H2}, respectively, of iteration $k$. Observe that every vertex $v \not\in F^\bullet$ has a non-neighbour $v' \in \borda{F^\circ}$ because otherwise $\borda{F^\circ} \cup \{v\}$ would be a clique, which would mean that $\borda{F^\circ} \cup \{v\}$ is contained in some mp-subgraph of $G$, which would imply that $v$ belongs to $F^\bullet$ by the construction of $F^\bullet$. Observe that Lemma~\ref{lem:concave}~$(\ref{ite:inter})$ implies $\borda{F^\circ} \not\subseteq \borda{F^\bullet}$ because otherwise it would exist an mp-subgraph outside $F^\bullet$ containing $\borda{F^\circ}$ which was not used to compose $F^\bullet$. Therefore, $\borda{F^\circ} - \borda{F^\bullet} \neq \varnothing$ and we can assume that $v' \in \borda{F^\circ} - \borda{F^\bullet}$ because for any set $T$, there is no edge between a vertex of $\semborda{T}$ and a vertex of $V(G) - T$.

For every set $T$ added to ${\cal F}$ such that $\semborda{T}$ has type~1, we associate a natural number $\ell(T)$. It is clear that a set is added to ${\cal F}$ at most once and this occurs in lines~\ref{lin:extremal} or~\ref{lin:join}. If $T$ is added to ${\cal F}$ in line~\ref{lin:extremal}, set $\ell(T) = 1$. If $T$ is added to ${\cal F}$ in line~\ref{lin:join} of iteration $k$ and $\semborda{F}$ is not t-concave for every $F \in {\cal F'}$, then set $\ell(T) = 1$. Otherwise, define $\ell(T) = 1 + \max \{ \ell(F) : F \in {\cal F'}$ and $\semborda{F}$ is t-concave$\}$. Lemma~\ref{lem:formed} guarantees that if $\semborda{F}$ is t-concave for $F \in {\cal F'}$ in iteration $k$, then the type of $\semborda{F}$ is~1.
It is clear that $\ell$ is well defined. We use induction on $\ell(F^\bullet)$ to prove that $\semborda{F^\bullet} \subset \langle u,w \rangle$. For the basis, consider $\ell(F^\bullet) = 1$. 

Suppose that Choice~\ref{cho:T1F0} is possible in line~\ref{lin:choT1k0}.
Let $P_{uw'}$ be a $(u,w')$-path of $F_1 - (\borda{F^\circ} - w')$, and
let $P'''$ be a $(w',u')$-path of $G[F_2 - (F_1 - \{w',u'\})]$.
It is clear that the paths $P_{wu'}$, $P_{uw'}$, and $P'''$ form a tolled $(u,w)$-walk.
For every $F' \in {\cal M'} \cup {\cal F'} - \{F_1,F_2\}$, it holds that $u'',w' \in \borda{F'}$ because $u'',w' \in \borda{F^\circ} \subseteq \borda{F'}$.
Observe that $N[w] \cap F' \subset \borda{F^\circ} \subseteq \borda{F'}$ and $N[u] \cap F' \subset \borda{F_1} \cap \borda{F'}$. Since $\borda{F^\circ} \subseteq \borda{F_1}$, $(N[u] \cup N[v]) \cap F' \subseteq \borda{F_1} \cap \borda{F'}$.
By Lemma~\ref{lem:F}~$(\ref{ite:FM})$, $\borda{F'} \cap \borda{F_1}$ is a clique, which means that $\borda{F'} \cap \borda{F_1}$ is contained in an mp-subgraph forming $F'$ and that every vertex of $F' \cap F_1$ has a neighbor in $F' - F_1$.
By Lemma~\ref{lem:F}~$(\ref{ite:maximal})$, $F' - F_1$ induces a connected graph. 
Then, there is a $(u'',w')$-walk containing all vertices of $F' - F_1$. Therefore, $F' - F_1 \subset \langle u,w \rangle$ for every $F' \in {\cal M'} \cup {\cal F'} - \{F_1,F_2\}$.

Now, we show that $F_2 - (F_1 \cup \borda{F^\bullet}) \subset \langle u,w \rangle$
Let $C_1, \ldots, C_s$ be the vertex sets of the connected components of $G[F_2] - (F_1 \cup \borda{F^\bullet})$. On the one hand, $\borda{F^\bullet}$ and $F_1 \cap F_2$ are incomparable sets. Since $G[F_2] - \borda{F^\bullet}$ and $G[F_2] - F_1$ are connected graphs,
for $1 \le z \le s$, $C_z$ contains 
a vertex $x_z$ having a neighbor in $\borda{F^\bullet} - \borda{F_1}$ and 
a vertex $y_z$ having a neighbor in $\borda{F_1} - \borda{F^\bullet}$.
By Lemma~\ref{lem:F}~$(\ref{ite:maximal})$,
there is a $(w,x_z)$-path $P_z$ of $G - F_1$ and
there is a $(u,y_z)$-path $P'_z$ of $G - \borda{F^\bullet}$.
Now, for each set $C_z$, the paths $P_z$ and $P'_z$ can be used to find a tolled $(u,w)$-walk such that, using Lemma~\ref{lem:walk}, we can conclude that $F_2 - (F_1 \cup \borda{F^\bullet}) \subset \langle u,w \rangle$.
On the other hand, $\borda{F^\bullet} \subset \borda{F_1}$. Therefore, $s = 1$ and we can assume that $u'$ and $w'$ belong to $\borda{F_1}$, and a tolled $(u,w)$-walk containing all vertices of $F_2 - F_1$ also exists. Then, $F_2 - (F_1 \cup \borda{F^\bullet}) \subset \langle u,w \rangle$.

Next, we show that $\semborda{F_1} \subset \langle u,w \rangle$.
We claim that for every $z \in \semborda{F_1}$, there are $x,y \not\in F_1$ for which there is a tolled $(x,y)$-walk containing $z$.
If $\borda{F_1}$ is clique, then, since $\borda{F_1}$ is contained in an mp-subgraph $M_1$ composing $F_1$, $G[\semborda{F_1}]$ is connected, and the follows from Lemma~\ref{lem:walk} and the fact that $\semborda{F_1}$ is not t-concave.
Then, consider that $\borda{F_1}$ is not a clique and let $C$ be a connected component of $\semborda{F_1}$. Then, there are non-adjacent vertices $x^*,y^* \in \borda{F_1}$ both having neighbors in $C$ and sets $F_x,F_y \in ({\cal F'} \cup {\cal M'}) - \{F_1\}$ such that $x^* \in F_x - F_y$ and $y^* \in F_y - F_x$.
Since we have already shown that $F_x - F_1 \subset \langle u,w \rangle$ and $F_y - F_1 \subset \langle u,w \rangle$, for any vertices $x \in F_x - F_1$ and $y \in F_y - F_1$, there is a tolled $(x,y)$-walk containing a vertex of $C$. Using Lemma~\ref{lem:walk} again the proof of the claim is complete.

If $x,y \in \semborda{F^\bullet} - F_1$, then we are done. The other possibility is that exactly one of $x$ and $y$, say $x$, does not belong to $\semborda{F^\bullet}$. Therefore, $y$ has a non-neighbor in $\borda{F^\bullet}$, which means that there is a tolled $(y,w)$-walk containing $z$, which implies that $\semborda{F_1} \subset \langle u,w \rangle$.

It remains to show that $\borda{F_1} - \borda{F^\bullet} \subset \langle u,w \rangle$. Let $z \in \borda{F_1} - \borda{F^\bullet}$. There is some $F' \in ({\cal F} \cup {\cal M}) - F_1$ such that $z \in F_1 \cap F'$. Then, $v$ has a neighbor in $F_1 - F'$ and a neighbor in $F' - F_1$. Since we know that $F_1 - F' \subset \langle u,w \rangle$ and $F' - F_1 \subset \langle u,w \rangle$, it holds that $z \in \langle u,w \rangle$.

We now consider that only Choice~\ref{cho:T1FJ} can be done.
On the one hand, $\borda{F^\bullet} \subseteq \borda{F}$ for every $F \in {\cal F'} \cup {\cal M'}$. For every $F' \in {\cal F'} \cup {\cal M'} - F_1$, note that $u', w' \in \borda{F'}$ and both have neighbors in $F' - F_1$. Therefore, there is a $(u',w')$-walk $W$ in $F' - (F_1 - \{u',w'\})$ containing a vertex $v \in F' - F_1$. The concatenation of $P_{uw'}$, $W$, and $P_{wu'}$ is a tolled $(u,w)$-walk containing $v$. 
By Lemma~\ref{lem:F}~$(\ref{ite:maximal})$, $G[F' - F_1]$ is connected. By Lemma~\ref{lem:walk}, $F' - F_1 \subset \langle u,w \rangle$. The same reasonings of the previous case can be used here to show that $\semborda{F_1} \subset \langle u,w \rangle$ and that $\borda{F_1} - \borda{F^\bullet} \subset \langle u,w \rangle$.

On the other hand, there is $F \in {\cal F'} \cup {\cal M'}$ such that $\borda{F^\bullet} \not\subseteq{F}$. Note that for every $F \in {\cal F'} \cup {\cal M'}$ such that $\borda{F^\bullet} \not\subseteq{F}$, $(F - \borda{F^\circ}, \borda{F^\circ})$ is complete. For instance, this is the case for $F^\circ$ and for the set chosen as $F_1$.
Since $\semborda{F^\circ}$ is not t-concave, there are vertices $x,y \not\in F^\circ$ for which there is a tolled $(x,y)$-walk containing some vertex $z \in \semborda{F^\circ}$.
Therefore, for $F' \in ({\cal F'} \cup {\cal M'}) - \{F^\circ\}$ such that $\borda{F^\bullet} \not\subseteq F'$, $x,y \not\in F'$. Furthermore, $x$ and $y$ have both non-neighbors in $\borda{F^\circ}$.
Without loss of generality, there is $F_x \in {\cal F'} \cup {\cal M'}$ such that $x \in F_x$. As in the previous case, $F_x - (F_1 \cup \borda{F^\bullet}) \subset \langle u,w \rangle$.
In fact, $F'' - (F_1 \cup \borda{F^\bullet}) \subset \langle u,w \rangle$ for every $F''$ such that $\borda{F^\bullet} \subseteq{F''}$. Since $x \not\in F_1$, $x \in F_x - F_1$, i.e., $x \in \langle u,w \rangle$. If $y \in \semborda{F^\bullet}$, then using the same reasoning, we conclude that $y \in \langle u,w \rangle$. Otherwise, consider $y = w$. Therefore, $z \in \langle u,w \rangle$. More generally, let $z' \in F$ such that $\borda{F^\bullet} \not\subseteq{F}$. Not that $z'$ dominates $\borda{F^\circ}$. Then, $xy'z'x'y$ is a $(x,y)$-walk containing a tolled $(x,y)$-walk passing through $z'$.
The same reasoning of the previous case can be used here to show that $\borda{F_1} - \borda{F^\bullet} \subset \langle u,w \rangle$ and that $\borda{F^\circ} - \borda{F^\bullet} \subset \langle u,w \rangle$.

Now, consider $\ell(F^\bullet) \geq 2$ and that the result holds for every t-concave set $F$ added to ${\cal F}$ such that $\ell(F) < \ell(F^\bullet)$. Hence ${\cal F'} - \{F^\circ\}$ contains at least one member $F$ such that $\semborda{F}$ is t-concave of type 1, say $F_1$. By the induction hypothesis, there is $u_1 \in S^* \cap \semborda{F_1}$ and $\semborda{F_1} \subset \langle u,w \rangle$.
We claim that some vertex of $\semborda{F_1}$ has a non-neighbor in $\borda{F^\circ}$. Suppose the contrary. Since $\semborda{F^\circ}$ is not t-concave, there exist vertices $z,z' \not\in F^\circ$ such that $\semborda{F^\circ} \cap [z,z'] \neq \varnothing$. Since both $z$ and $z'$ have at least one non-neighbor in $\borda{F^\circ}$ and, by assumption, $\semborda{F^\bullet}$ is t-concave, at least one, say $z$, belongs to $\semborda{F^\bullet} - \semborda{F^\circ}$. Since every vertex of $\borda{F^\circ}$ has a neighbor in $F_1 - F^\circ$, we conclude that a tolled $(z,z')$-walk $W$ containing some vertex of $\semborda{F^\circ}$ can be modified to contain a vertex of $\semborda{F_1}$, which is not possible because $\semborda{F_1}$ is t-concave.
Therefore, let $u \in \semborda{F_1}$ having a non-neighbor $u'' \in \borda{F^\circ}$. Recall that $w' \in \borda{F^\circ}$ is a vertex non-adjacent to $w$. Therefore, there is a tolled $(u,w)$-walk $W'$ containing vertices $u''$ and $w'$. 
Note that for $F' \in {\cal M'} \cup {\cal F'} - \{F_1\}$, $F' - (F_1 - w')$ induces a connected graph by Lemma~\ref{lem:F}~$(\ref{ite:maximal})$.
Hence, there is a walk containing all vertices of $F' - F_1$ that can be concatenated with $W'$ forming a tolled $(u,w)$-walk containing all vertices of $F' - F_1$.
Therefore, $F' - F_1 \subset \langle u,w \rangle = \langle u_1,w \rangle$ for every $F' \in {\cal M'} \cup {\cal F'} - \{F_1\}$.
It remains to show that $\borda{F_1} - \borda{F^\bullet} \subset \langle u,v \rangle$.
Let $z \in \borda{F_1} - \borda{F^\bullet}$. This means that there is $F \in {\cal F'} \cup {\cal M'} - F_1$ such that $z \in \borda{F}$, and consequently that $z$ has a neighbor in $F_1 - F$ and a neighbor in $F - F_1$. Since $\semborda{F_1} \subset \langle u,w \rangle$ and $F - F_1 \subset \langle u,w \rangle$, we conclude that $\borda{F_1} - \borda{F^\bullet} \subset \langle u,v \rangle$. Therefore, $\semborda{F^\bullet} \subset \langle u,w \rangle$.
$\hfill\square$
\end{proof}

The above proof has the following consequence.

\begin{corollary} \label{cor:T12}
	If $F^\bullet$ was obtained in line~$\ref{lin:F3}$ of Algorithm~$\ref{alg:thn}$ such that $\semborda{F^\bullet}$ has type~$1$, then the family ${\cal F'}$ obtained in line~$\ref{lin:H2}$ of the same iteration has at most one member $F$ such that $\semborda{F}$ is t-concave.
\end{corollary}

Now, we show that the toll hull number of $C \in {\cal C}$ is $2$ if $C$ has type~2.

\begin{lemma} \label{lem:2T2}
If $F \in {\cal F^*}$ is such that $\semborda{F}$ is t-concave of type~$2$, then $S^* \cap \semborda{F} = \{u_1,u_2\}$ and $\semborda{F} \subseteq \langle u_1,u_2 \rangle$ for $u_1,u_2 \in \semborda{F}$.
\end{lemma}

\begin{proof}
Let $F^\bullet \in {\cal F^*}$ such that $\borda{F^\bullet}$ is t-concave of type~$2$.
If $F^\bullet$ was added to $\cal F$ in line~\ref{lin:extremal}, then the result follows from Corollary~\ref{cor:atom} because $F^\bullet$ is an mp-subgraph.
Now, we consider that $F^\bullet$ was added to $\cal F$ in line~\ref{lin:join} of some iteration $k$ of the {\bfseries while} loop. Let $\semborda{F^\circ}, {\cal M'}$, and ${\cal F'}$ be obtained in lines~\ref{lin:F}, ~\ref{lin:H}, and~\ref{lin:H2} of iteration $k$, respectively.
If $\borda{F^\bullet} - \borda{F^\circ} \ne \varnothing$, then it would exist a vertex of $\semborda{F^\bullet}$ having a non-neighbor in $\borda{F^\bullet}$, which would imply that the type of $\semborda{F^\bullet}$ is not 2. Then, $\borda{F^\bullet} \subset \borda{F^\circ}$.
By Lemma~\ref{lem:formed}, the number $j$ of sets $F$ of ${\cal F'}$ such that $\semborda{F}$ is t-concave is at most~$2$.

If $j = 0$ and Choice~\ref{cho:T20T1a} was done in line~\ref{lin:choT2k0}, then, by Lemma~\ref{lem:walk}, it holds $F' - \borda{F^\circ} \subset \langle u_1,u_2 \rangle$ for every $F' \in {\cal M'} \cup {\cal F'} - \{F_1\}$ because $F' - (\borda{F^\circ} - u_i)$ induces a connected graph for $i \in \{1,2\}$ by Lemma~\ref{lem:F}~$(\ref{ite:maximal})$.
If $j = 0$ and Choice~\ref{cho:T20T1b} was done in line~\ref{lin:choT2k0}, or $j = 1$ and Choice~\ref{cho:T21T1} was done in line~\ref{lin:F*end}, or $j = 2$, then, by Lemma~\ref{lem:walk}, it holds $F' - \borda{F^\circ} \subset \langle u_1,u_2 \rangle$ for every $F' \in {\cal M'} \cup {\cal F'} - \{F_1,F_2\}$ because $F' - \borda{F^\circ}$ induces a connected graph by Lemma~\ref{lem:F}~$(\ref{ite:maximal})$. If $j = 1$ and Choice~\ref{cho:T21T2} was made, then $F - F_1 \subset \langle u_1,u_2 \rangle$ for every $F \in {\cal F'} \cup {\cal M'} - \{F_1\}$.

Therefore, it remains to show that $\semborda{F_1} \subseteq \langle u_1,u_2 \rangle$ for Choice~\ref{cho:T20T1a}~or~\ref{cho:T21T2} and that $\semborda{F_1} \cup \semborda{F_2} \subseteq \langle u_1,u_2 \rangle$ for $j = 2$ or Choice~\ref{cho:T20T1b} or~\ref{cho:T21T1}.

First consider $j = 0$. For Choice~\ref{cho:T20T1a}, since no vertex outside $F^\bullet$ belongs to a tolled walk containing vertices of $\semborda{F^\bullet}$, there are vertices $w,w' \in \semborda{F^\bullet} - \semborda{F_1}$ whose tolled interval contains vertices of $\semborda{F_1}$. Since we have already shown that $\semborda{F^\bullet} - \semborda{F_1} \subset \langle u_1,u_2 \rangle$, we have $\semborda{F_1} \subset \langle u_1,u_2 \rangle$

For Choice~\ref{cho:T20T1b}, we can assume that Choice~\ref{cho:T20T1a} is not possible. Therefore, for some $F \in {\cal M'} \cup {\cal F'} - \{F_1,F_2\}$, there is $w \in F$ such that $\semborda{F_1} \subset \langle u_2,w \rangle$ or there are $w \in F$ and $w' \in F'$ such that $\semborda{F_1} \subset \langle w,w' \rangle$. For both cases, we have that $\semborda{F_1} \subset \langle u_1,u_2 \rangle$ because we have already shown that $\semborda{F^\bullet} - \semborda{F_1} \subset \langle u_1,u_2 \rangle$. It is clear that the same does hold for $\semborda{F_2}$.

For $j = 1$, consider first that Choice~\ref{cho:T21T1} is possible. From Lemma~\ref{lem:1T1}, it holds that $\semborda{F_1} \subseteq \langle u_1,u_2 \rangle$. Now, since no vertex outside $F^\bullet$ belongs to a tolled walk containing vertices of $\semborda{F_2}$, it holds $\semborda{F_2} \subset \langle \semborda{F^\bullet} - F_2 \rangle$, and we have already shown that $\semborda{F^\bullet} - F_2 \subset \langle u_1,u_2 \rangle$.
For Choice~\ref{cho:T21T2}, From Lemma~\ref{lem:1T1}, it holds that $\semborda{F_1} \subseteq \langle u_1,u_2 \rangle$.

For $j = 2$, the result follows directly from Lemma~\ref{lem:1T1}. 

It remains to show that $(\borda{F_1} \cup \borda{F_2}) - \borda{F^\bullet} \subset \langle u_1,u_2 \rangle$.
Let $w_1 \in \borda{F_1} - \borda{F^\bullet}$.
We know that $w_1$ also belongs to a set $F' \in {\cal M'} \cup {\cal F'} - \{F_1\}$. If $F' = F_2$, then 
observe that $w_1$ belongs to a tolled $(u_1,u_2)$-walk. Otherwise, 
observe that $w_1$ belongs to a tolled $(u_1,w')$-walk where $w' \in F' - F_1$. Since $w' \in \langle u_1,u_2 \rangle$, $\semborda{F^\bullet} \subseteq \langle u_1,u_2 \rangle$.
$\hfill\square$
\end{proof}

\begin{theorem} \label{thm:hrf}
	Algorithm~{\em \ref{alg:thn}} is correct.
\end{theorem}

\begin{proof}
For every set $F \in {\cal F^*}$ such that $\semborda{F}$ is t-concave, if $\semborda{F}$ has type~$i$ for $i \in \{1,2\}$, then $|S^* \cap \semborda{F}| = i$; and if $i = 3$, then $\semborda{F} \subseteq S^*$. Since every vertex of $S^*$ belongs to $\semborda{F}$ for some $F \in {\cal F^*}$ such that $\semborda{F}$ is t-concave and $\semborda{F} \cap \semborda{F'} = \varnothing$ for any $F,F' \in {\cal F^*}$, by Lemma~\ref{lem:granularity}, it suffices to show that $S^*$ is a toll hull set of $G$. We also have as a consequence that $\cal C$ is a toll hull characteristic family of $G$.

First, consider ${\cal F^*} = \{F\}$. Then, $\borda{F} = \varnothing$, which means that the type of $\semborda{F}$ is 2 or 3. In the former case, the result follows from Lemma~\ref{lem:2T2} and, in the latter case, $G$ is a complete graph and the result follows from Corollary~\ref{cor:atom}.
	
Now, consider $|{\cal C}| \geq 2$ and let $F \in {\cal F^*}$ such that $\semborda{F}$ is t-concave.
First, we show that $\semborda{F} \subset \langle S^* \rangle$. If $\semborda{F}$ has type 3, then $\semborda{F} \subseteq S^*$ by line~\ref{lin:mp-subgraph-end}. If $\semborda{F}$ has type 2, then $\semborda{F} \subset \langle S^* \rangle$ by Lemma~\ref{lem:2T2}. In the remaining case, $\semborda{F}$ has type 1. By Lemma~\ref{lem:1T1}, it suffices to show that $V(G) - F$ contains a set of ${\cal C}$. But this is a consequence of Lemma~\ref{lem:F}~$(\ref{ite:disjoint})$ and the fact that $|{\cal C}| \geq 2$.

It remains to show that every $v \not\in \underset{C \in {\cal C}}{\bigcup} C$ belongs to $\langle S^* \rangle$. Suppose the contrary and let $v \in B \subset V(G) - \langle S^* \rangle$ such that $G[B]$ is connected and $B$ is maximal. Write $G' = G - B$. Note that $v \in B$ belongs to a member of ${\cal M^*}$ or to a member $F \in {\cal F^*}$ such that $\semborda{F}$ is not t-concave.
In the former case, $G'$ is disconnected by Corollary~\ref{cor:non-extremal}.
In the latter case, by the assumption on ${\cal F^*}$, $\borda{F}$ is not contained in any another member. Then, Lemma~\ref{lem:F}~$(\ref{ite:border-not-contained})$ implies that $G'$ is disconnected.
In both cases, each of the connected components contains a member of ${\cal C}$, i.e., there are vertices $u_1, u_2 \in S^*$ belonging to different connected components of $G'$. It is clear that any tolled $(u_1,u_2)$-walk contains some vertex of $B$, which is a contradiction.
$\hfill\square$
\end{proof}

The task of deciding whether a set is t-concave appears many tymes in Algorithm~\ref{alg:thn}. A direct application of Lemma~\ref{lem:tollconvexset} leads to an algorithm for deciding whether a set is t-concave with time complexity equals $O(n^3m)$. However, all sets that this test is necessary in the algorithm have particular properties. In the next result, we show how to decide whether these sets are t-concave in $O(n^3)$ steps.

\begin{lemma} \label{lem:tconvex}
Let $G$ be a graph with $n$ vertices. If $F \subseteq V(G)$ is such that $\borda{F}$ is a clique and $G[\semborda{F}]$ is connected, then to test whether $\semborda{F}$ is t-concave can be done in $O(n^3)$ steps.
\end{lemma}

\begin{proof}
By definition, $\semborda{F}$ is not t-concave if and only if there are $x,y \not\in \semborda{F}$ such that there is a tolled $(x,y)$-walk containing some vertex of $\semborda{F}$. The assumption that $\borda{F}$ is a clique implies that $x,y \not\in F$. Since $\semborda{F}$ is connected, by Lemma~\ref{lem:walk}, either all vertices of $\semborda{F}$ belong to $[x,y]$ or none. Let $v \in \semborda{F}$. Since $v \not\in N(x) \cup N(y)$, by Lemma~\ref{lem:tollconvexset}, $v$ is in some tolled walk between non-adjacent vertices $x,y \not\in F$ if and only if $N[x]$ does not separate $v$ from $y$ and $N[y]$ does not separate $v$ from $x$.

For every $u \not\in F$, denote by ${\cal T}_u$ the family formed by the vertex sets of the connected components of $G - N[u]$. Using an standard search algorithm, all families can be found in $O(nm)$ steps.
Now, observe that $v$ belongs to $[V(G) - F]$ if and only if there are $u,z \not\in F$ with $uz \not\in E(G)$ such that $v \in {\cal T}_u(z) \cap {\cal T}_z(u)$. Since this test can be done in $O(n^3)$ steps, the proof is complete.
$\hfill\square$
\end{proof}

\begin{theorem}
For an input graph with $n$ vertices, Algorithm~$\ref{alg:thn}$ runs in $O(n^4)$ steps.
\end{theorem}

\begin{proof}
Lines~\ref{lin:prime} and~\ref{lin:decomposition} can done in $O(nm)$ using the algorithm in~\cite{Leimer1993}.
The construction of ${\cal F}$ and ${\cal M}$ in lines~\ref{lin:extremal} and~\ref{lin:othermp} can be done in $O(n^3)$.
The number of iterations of the {\bfseries for} loop is $O(n)$. Using Lemma~\ref{lem:tconvex}, one can test whether a set is t-concave in $O(n^3)$ steps. Since the time complexity of each Choice $i$ for $i \in \{\ref{cho:T1}, \ldots, \ref{cho:T21T2}\}$ is clearly $O(n^2)$, lines~\ref{lin:mp-subgraph-beg} to~\ref{lin:mp-subgraph-end} can be done in $O(n^2m)$ steps.
	
Every time that line~\ref{lin:F} is reached, we already know for each member of $F \in \cal F$, whether $\semborda{F}$ is not t-concave. The conditions of line~\ref{lin:F} can be checked in $O(n^3)$ considering the intersection of the mp-subgraphs forming the members of ${\cal F} \cup{\cal M}$.
Clearly, each operation from line~\ref{lin:H} to~\ref{lin:join} can be done in $O(n^2)$.
Since Line~\ref{lin:concave} can be done in $O(n^3)$ time by Lemma~\ref{lem:tconvex} and lines~\ref{lin:F*beg} to~\ref{lin:F*end} can be done in $O(n^2)$, the {\bfseries while} loop costs $O(n^4)$, which is the overall time complexity of Algorithm~\ref{alg:thn}.
$\hfill\square$
\end{proof}

%%%%%%%%%%%%%%%%%%%%%%%%%%%%%%%%%%%%%%%%%%%%%%%%%
\section{Concluding remarks}

We conclude discussing some consequences of Algorithm~\ref{alg:thn}. First, we observe that the number of minimum toll hull sets can be exponential on the size of the graph. However, using the toll hull characteristic family constructed by Algorithm~\ref{alg:thn}, one can enumerate all minimum toll hull sets of $G$ with polynomial time delay. For this, it suffices to change the choices used by the algorithm so that they find all possible selections for a concave set $C$ accordingly to the appropriate choice, i.e.,
if $C$ has type 1, let $t(C)$ be formed by all vertices $x$ such that $x$ satisfies the appropriate choice for $C$; and 
if $C$ has type 2, let $t(C)$ be formed by all pairs $\{x,y\}$ such that $\{x,y\}$ satisfies the appropriate choice for $C$. Therefore, the algorithm of enumaration consists of finding all combinations considering the possible choices for each concave set of the toll hull characteristic family.

Another consequence of Algorithm~\ref{alg:thn} together with the notion of granularity is a characterization of toll extreme vertices of a graph. As discussed in~\cite{Alcon2015}, the property of a vertex being an extreme vertex is not well-behaviored in toll convexity as in other well-studied convexities, such as geodetic, nonophonic, and $P_3$ convexities, where the neighborhood of the vertex has all information to answer the question. Using the toll hull characteristic family of Algorithm~\ref{alg:thn}, we have the following characterization of the toll extreme vertices of a graph.

\begin{corollary}
The set of toll extreme vertices of $G$ is formed by the vertices belonging to the sets of ${\cal C}$ having type~$3$.
\end{corollary}

\begin{proof}
Let $X$ be formed by the vertices belonging to the sets of ${\cal C}$ having type~$3$. Suppose by contradiction that there is $u \not\in X$ such that $u$ is an extreme vertex. Since every non-simplicial vertex is not an extreme vertex, $N(u)$ is a clique. Observe that $N[u]$ is the vertex set of an mp-subgraph $F$ of $G$ and $u \in \semborda{F}$. On the one hand, $F \in {\cal F^*} \cup {\cal M^*}$. Since $\semborda{F} \not\in {\cal C}$, $F \in {\cal M^*}$. By Corollary~\ref{cor:non-extremal}, $G - F$ is disconnected, which means that $u$ is not an extreme vertex. On the other hand, $F$ is an mp-subgraph composing a set $F' \in {\cal F^*} \cup {\cal M^*}$. Since $u$ is an extreme vertex, $\semborda{F'}$ is t-concave. Therefore, in some iteration of the {\bfseries while} loop, $\borda{F^\circ} \subseteq \borda{F}$ such that $\borda{F^\circ}$ is not t-concave. Then, there are of vertices $v,w \not\in F^\circ$ such that there is a tolled $(v,w)$-walk $W$ containing some vertex of $\semborda{F^\circ}$. Since $(\semborda{F},\borda{F})$ is complete, $v,w \not\in F$. Let $v'$ and $w'$ vertices of $W$ such that $vv',ww' \not\in E(G)$. Denoting
$W_{vw'}$ the subwalk of $W$ from $v$ to $w'$ and
$W_{v'ww}$ the subwalk of $W$ from $v'$ to $w$, 
Note that the concatenation of $W_{vw'}$, $vxv'$, and $W_{v'w}$ is a tolled $(v,w)$-walk containing $u$, which is a contradiction.
$\hfill\square$
\end{proof}



\begin{thebibliography}{}

	\bibitem[Albenque and Knauer 2016]{albenque}
	Albenque, M., \and Knauer K. (2016).
	Convexity in Partial Cubes: The Hull Number,
	{\it Discrete Mathematics},
	339, 866--876.
	
	\bibitem[Alcón et al. 2015]{Alcon2015}
	Alcón, L., Brešar, B., Gologranc, T., Gutierrez, M., Šumenjak,T.K., Peterin,I., \and Tepeh,A. (2015),
	Toll convexity,
	{\it European Journal of Combinatorics}
	46, 161--175.
	
	\bibitem[Araujo et al. 2013]{araujoTCS}
	Araujo, J., Campos, V., Giroire, F.,  Nisse, N., Sampaio, L., \and Soares, R. (2013).
	On the hull number of some graph classes,
	{\it Theoret. Comput. Sci.},
	475, 1--12.
	
	\bibitem[Araujo et al. 2016]{araujoENDM}
	Araujo, J., Morel, G., Sampaio, L., Soares, R., \and Weber, V. (2016).
	Hull number: $P_5$-free graphs and reduction rules,
	{\it Discrete Applied Mathematics}
	210, 171--175.
	
	\bibitem[Bessey et al. 2018]{BDPR2018}
	Bessy, S., Dourado, M.C., Penso, L.D., Rautenbach, D. (2018).
	The Geodetic Hull Number is Hard for Chordal Graphs,
	{\it SIAM Journal on Discrete Mathematics}, 
	32, 543--547. 
	
	\bibitem[Centeno et al. 2011]{Centeno2011}
	Centeno, C.C., Dourado, M.C., Penso, L.D., Rautenbach, D., Szwarcfiter, J.L. (2011).
	Irreversible conversion of graphs,
	{\em Theoretical Computer Science},
	412, 3693--3700.
	
	\bibitem[Coelho et al. 2015]{CDS2015}
	Coelho, E.M.M., Dourado,M.C., Sampaio, R.M. (2015).
	Inapproximability results for graph convexity parameters,
	{\em Theoretical Computer Science},
	600, 49--58.
	
	\bibitem[Dourado et al. 2009]{DGKPS2009}
	Dourado, M.C., Gimbel, J.G., Kratochvíl, J., Protti, F., Szwarcfiter, J.L. (2009).
	On the computation of the hull number of a graph,
	{\em Discrete Mathematics},
	309, 5668--5674.
	
	\bibitem[Dourado et al. 2010]{DPS-2010-monophonic}
	Dourado, M.C., Protti, F., Szwarcfiter, J.L. (2010).
	Complexity results related to monophonic convexity,
	{\it Discrete Applied Mathematics}
	158:12, pp. 1268--1274.
	
	\bibitem[Dourado et al. 2012]{upper-Radon}
	Dourado, M.C., Rautenbach, D., dos Santos, V.F., Sch{\"{a}}fer, P.M., Szwarcfiter, J.L., Toman, A. (2012).
	An upper bound on the $P_3$-Radon number,
	{\em Discrete Mathematics},
	312, 2433--2437.
	
	\bibitem[Duchet 1988]{Duchet-mono}
	Duchet, P. (1988).
	Convex sets in graphs, II: minimal path convexity
	{\it J. Comb. Theory, Ser. B}
	44, 307--316.
	
	\bibitem[Edelman and Jamison 1985]{EJ1985}
	Edelman, P.H., Jamison, R.E. (1985).
	The theory of convex geometries,
	{\it Geometriae Dedicata},
	19, 247--270.
	
	\bibitem[Farber and Jamison 1986]{FarberJamison1986}
	Farber, M., Jamison, R.E. (1986).
	Convexity in graphs and hypergraphs,
	{\it SIAM J. Alg. Disc. Meth.},
	7, 433--444.
	
	\bibitem[Gimbel 2003]{Gimbel2003}
	Gimbel, J.G. (2003).
	Some remarks on the convexity number of a graph,
	{\it Graphs Comb.},
	19, 357--361.
	
	\bibitem[Gologranc and Repolusk 2017]{GR2017}
	Gologranc, T., Repolusk, P. (2017).
	Toll number of the Cartesian and the lexicographic product of graphs,
	{\it Discrete Mathematics},
	340, 2488--2498.
	
	\bibitem[Henning et al. 2013]{Henning2013}
	Henning, M.A., Rautenbach, D., Schäfer, P.M. (2013).
	Open packing, total domination, and the $P_3$-Radon number,
	{\em Discrete Mathematics},
	313, 992--998.
	
	\bibitem[Kante and Nourine 2016]{KN2016}
	Kante, M.M., Nourine, L. (2016).
	Polynomial time algorithms for computing a minimum hull set in distance-hereditary and chordal graphs,
	{\it SIAM J. Discrete Math.},
	30(1), 311--326.
	
	\bibitem[Leimer 1993]{Leimer1993}
	Leimer, H.-G. (1993). 
	Optimal Decomposition by clique separators,
	{\em Disc. Math.},
	113, 99--123.
	
	\bibitem[Pelayo 2013]{pelayo2013}
	Pelayo, I.M. (2013).
	{\it Geodesic Convexity in Graphs},
	Springer.
	
	\bibitem[van de Vel 1993]{Vel1993}
	van de Vel, M.L.J. (1993).  
	{\it Theory of Convex Structures},
	North-Holland, Amsterdam.
	
\end{thebibliography}
\end{document}